\documentclass[twocolumn,superscriptaddress,aps,pra]{revtex4-2}

\usepackage{graphicx}
\usepackage{amsmath}
\usepackage{amsthm}
\usepackage{amssymb}
\usepackage{latexsym}
\usepackage{array}
\usepackage{hyperref}
\usepackage{float}
\usepackage{amsfonts}
\usepackage{dsfont}
\usepackage{mathrsfs}
\usepackage{verbatim}
\usepackage{bbold}
\usepackage[normalem]{ulem}
\usepackage{upgreek}
\usepackage{makecell}
\usepackage{adjustbox}
\usepackage{color}

\usepackage{bm}
\usepackage{times}

\newcommand{\bra}[1]{\left<#1\right|}
\newcommand{\ket}[1]{\left|#1\right>}
\newcommand{\abs}[1]{\left|#1\right|}

\newcommand{\braket}[2]{\left<{#1}|{#2}\right>}
\newcommand{\ketbra}[2]{\ket{#1}\!\!\bra{#2}}

\newtheorem{theorem}{Theorem}
\newtheorem{proposition}{Proposition}

\newtheorem{corollary}{Corollary}
\newtheorem{definition}{Definition}
\newtheorem{remark}{Remark}

\DeclareMathOperator{\Tr}{Tr}
\DeclareMathOperator{\Var}{Var}

\newcommand{\hD}{\hat{D}}
\newcommand{\ha}{\hat{a}}
\newcommand{\had}{\hat{a}^{\dagger}}
\newcommand{\hq}{\hat{q}}
\newcommand{\hp}{\hat{p}}

\newcommand{\cE}{\mathcal{E}}

\newcommand{\C}{\mathbb{C}}
\newcommand{\R}{\mathbb{R}}

\begin{document}

\title{Universality cost of non-Gaussian enhancement in continuous-variable quantum teleportation: A fidelity--deviation trade-off}

\author{Kyoungho~Cho}
\affiliation{Institute for Convergence Research and Education in Advanced Technology, Yonsei University, Seoul 03722, Republic of Korea}
\affiliation{Department of Statistics and Data Science, Yonsei University, Seoul 03722, Republic of Korea}

\author{Bongjune~Kim}\email{bongjunekim@jejunu.ac.kr}
\affiliation{Department of Physics, Jeju National University, Jeju 63243, Republic of Korea}

\author{Jeongho~Bang}\email{jbang@yonsei.ac.kr}
\affiliation{Institute for Convergence Research and Education in Advanced Technology, Yonsei University, Seoul 03722, Republic of Korea}
\affiliation{Department of Quantum Information, Yonsei University, Incheon 21983, Republic of Korea}

\date{\today}

\begin{abstract}
Continuous-variable (CV) quantum teleportation is usually benchmarked by average fidelity, but when the teleportation is repeatedly used within optical networks or measurement-based architectures, uniformity across the input ensemble becomes equally important. We analyze this issue using two complementary figures of merit: the average fidelity and the fidelity deviation, which quantifies the input dependence of the single-shot teleportation fidelity. We prove that any deterministic unity-gain teleportation channel that is displacement covariant has vanishing fidelity deviation for coherent-state benchmarking, irrespective of whether the shared entangled resource is Gaussian or non-Gaussian. Nonzero deviation therefore diagnoses covariance breaking rather than non-Gaussianity. We then show that when a protocol raises the average fidelity through input-selective conditioning, the deviation generically increases in tandem, giving a quantitative universality cost. As a concrete example, we study teleportation enhanced by the so-called measurement-based noiseless linear amplification, where a heralded filter acts on the Bell-measurement record. The resulting trade-off among average fidelity, fidelity deviation, and success probability shows that stronger filtering can improve the conditional fidelity only by concentrating the successful events in favored regions of phase space, thereby suppressing the success probability and reducing input uniformity. Our results provide an operational framework for distinguishing genuine channel improvement from selectivity-driven post-selected advantage and suggest that the probabilistic CV teleportation should be assessed with average quality, universality, and heralding rate treated on an equal footing.
\end{abstract}

\maketitle

\section{Introduction}\label{sec:intro}

Continuous-variable (CV) quantum teleportation provides a conceptually simple and experimentally mature route to transferring an unknown state by combining shared entanglement, a CV Bell-measurement, and classical feed-forward displacement~\cite{BraunsteinKimble1998,Furusawa1998}. In its canonical Braunstein--Kimble (BK) form, the protocol is deterministic and Gaussian when the shared resource is a two-mode squeezed vacuum (TMSV). This ``Gaussian determinism'' is not merely a technical convenience. Indeed, it is one reason why CV teleportation has become a standard primitive in optical quantum networks, entanglement swapping, and measurement-based architectures, where teleportation is repeatedly invoked as a wiring operation~\cite{Menicucci2006}.

The teleportation performance is most often evaluated by the average fidelity over an input ensemble (typically, coherent states drawn from a Gaussian prior), because the teleportation is intended to work on unknown inputs and must be assessed statistically. Yet, average fidelity alone is an incomplete descriptor when the teleportation is used as a repeatedly composed subroutine. In modular protocols, it matters not only ``how well'' the teleportation works on average, but also ``how evenly'' it performs across the relevant input set: a procedure that is excellent for a narrow subset of inputs but poor elsewhere can be brittle under composition, can amplify input-distribution mismatch, and can conceal failure modes that remain invisible to mean-only figures of merit. This operational viewpoint naturally promotes a two-parameter characterization in terms of the single-shot fidelity profile $f(\alpha)$ (for coherent input $\ket{\alpha}$) and its low-order statistics over the chosen ensemble.

A particularly direct quantifier of input-uniformity is the fidelity deviation $D$, defined as the standard deviation of $f(\alpha)$ over the benchmarking ensemble. The fidelity deviation was introduced in finite-dimensional teleportation as a diagnostic of ``universality'' beyond mean performance~\cite{Bang2012,BangRyuKaszlikowski2018}, and its CV counterpart has been investigated in the context of coherent-state teleportation benchmarks~\cite{Patra2022}. Conceptually, the pair $(F, D)$ disentangles two questions that are otherwise conflated: the mean performance $F$ answers ``how good on average,'' while $D$ answers ``how sensitive to the input.'' This separation becomes especially meaningful because many physically motivated ``improvements'' of teleportation are not uniform channel improvements but rather selective procedures that reshape the distribution of successful events.

This brings us to a recurring theme in CV quantum information: the role of non-Gaussianity. While Gaussian states and Gaussian operations enjoy an exceptional analytical and experimental tractability, they are also constrained by no-go theorems: Gaussian operations alone cannot distill Gaussian entanglement~\cite{EisertScheelPlenio2002,GiedkeCirac2002}, and they cannot implement genuine Gaussian error correction against Gaussian noise~\cite{NisetFiurasekCerf2009}. Consequently, essentially every long-distance or fault-tolerant CV blueprint must introduce non-Gaussian elements at some point. In teleportation, non-Gaussianity can enter through resource engineering (for example, photon-subtracted or otherwise de-Gaussified entanglement), which has long been known to enhance average fidelity in suitable regimes~\cite{Opatrny2000,DellAnno2007}. More recently, a complementary and experimentally powerful route has been developed, called measurement-based noiseless linear amplification (MB-NLA), where one filters and/or post-processes the continuous Bell-measurement outcomes to realize a heralded, effectively non-Gaussian transformation that can boost the teleportation efficacy~\cite{Chrzanowski2014,Zhao2023,Fiurasek2024}. These schemes are intrinsically probabilistic and introduce an additional operational axis---the success probability---which in practice controls the repetition overhead and the effective resource consumption.

The central observation of this work is that the non-Gaussian enhancement is not synonymous with uniform channel improvement. In fact, whether a given non-Gaussian ingredient preserves or breaks the fundamental symmetries of the BK wiring determines whether the fidelity enhancement comes ``for free'' or comes at a price in the universality. The deterministic BK-type teleportation implements a displacement-covariant channel: displacing the input simply displaces the output by the same amount. For coherent-state benchmarking, this has a striking implication: the displacement covariance forces the single-shot fidelity profile to be flat in $\alpha$, and therefore implies $D=0$ irrespective of whether the shared entangled resource is Gaussian or non-Gaussian. In other words, within the coherent-state benchmarking, nonzero fidelity deviation is not a generic signature of non-Gaussian resources; rather, it is a direct witness that the implemented (possibly, heralded) procedure breaks the displacement covariance---through gain mismatch, phase-sensitive imperfections, or, most prominently for our purposes, through non-Gaussian conditioning that accepts and rejects events in a manner correlated with the measurement record and hence with the input.

This symmetry perspective suggests a concrete and testable story for the ``non-Gaussian advantage'' reported in the heralded CV teleportation. A post-selected filter can increase the average fidelity $F$ by preferentially retaining events associated with ``easier'' regions of phase space, or equivalently, by suppressing those measurement outcomes that would otherwise produce large displacements and hence large reconstruction errors. Such selectivity inevitably makes the successful teleportation map less uniform across inputs, thereby increasing the fidelity deviation $D$. The resulting phenomenon is a universality cost: improvements in $F$ driven by covariance breaking are accompanied by a predictable growth of $D$, and must be interpreted jointly with the heralding probability.

In this paper, we develop this universality-cost viewpoint into a quantitative framework. We first formalize the figures of merit $(F, D)$ for coherent-state ensembles in both deterministic and probabilistic teleportation, including the natural success-probability reweighting in the heralded case. We then prove a covariance--universality theorem showing that any deterministic displacement-covariant teleportation channel yields $D=0$ for the coherent-state benchmarking, thereby identifying $D>0$ as a clean diagnostic of the covariance breaking. Building on this, we derive a perturbative trade-off relation for small protocol deformations: when $F$ is increased by introducing an input-selective component to the fidelity profile, the induced deviation scales linearly with the fidelity gain, $D \simeq c \Delta F$, where the coefficient $c$ is determined by the variance of the selectivity profile over the effective (possibly, post-selected) ensemble. Finally, we apply the framework to MB-NLA-enhanced teleportation~\cite{Chrzanowski2014,Zhao2023,Fiurasek2024}, mapping the operational trade-off surface in the $(F,D,P_{\mathrm{succ}})$ space and identifying the regimes in which the non-Gaussianity raises $F$ primarily by concentrating performance on a restricted subset of inputs, thereby enlarging $D$.

\section{Preliminaries: CV teleportation and figures of merit}\label{sec:setting}

\subsection{Teleportation as an effective channel: deterministic vs filtered}\label{subsec:channel}

We consider a single bosonic mode with annihilation operator $\ha$ and canonical quadratures
\begin{eqnarray}
\hq=\frac{1}{\sqrt{2}}\left(\ha+\had\right), \quad \hp=\frac{1}{i\sqrt{2}}\left(\ha-\had\right),
\end{eqnarray}
so that $[\hq, \hp]=i$ (we set $\hbar=1$). The phase-space displacements are implemented by
\begin{eqnarray}
\hD(\alpha)=\exp\left(\alpha \had-\alpha^{\ast}\ha\right), \quad (\alpha \in \C),
\end{eqnarray}
and coherent states are $\ket{\alpha}=\hD(\alpha)\ket{0}$.

The Braunstein--Kimble (BK) protocol teleports an unknown input state $\hat{\rho}_{\mathrm{in}}$ by combining a shared bipartite resource $\hat{\rho}_{AB}$, a continuous-variable (CV) Bell-measurement on Alice's side, and classical feed-forward displacement on Bob's side~\cite{BraunsteinKimble1998,Furusawa1998,BraunsteinvanLoock2005}. It is convenient to describe the protocol at the level of a quantum instrument indexed by the Bell-measurement outcome. Denoting the complex Bell record by $m \in \C$ (equivalently, a pair of real homodyne outcomes), which in the standard BK circuit packages the outcomes of two commuting quadratures (often written as $q_{-}$ and $p_{+}$) measured after a balanced beam splitter, we write the corresponding completely positive (but generally trace-decreasing) map as $\cE_{m}$. Operationally, $\cE_{m}$ includes both the measurement update and Bob's conditional displacement. The (unnormalized) conditional output is $\cE_{m}(\hat{\rho}_{\mathrm{in}})$, while the probability density of obtaining $m$ is
\begin{eqnarray}
p(m | \hat{\rho}_{\mathrm{in}}) = \Tr\left[\cE_{m}(\hat{\rho}_{\mathrm{in}})\right],
\quad 
\int_{\C} d^{2}m \, p(m | \hat{\rho}_{\mathrm{in}})=1,
\label{eq:m_density}
\end{eqnarray}
where $d^{2}m$ denotes the Lebesgue measure on $\C$.

{\em Deterministic (Gaussian) teleportation.}---In the usual deterministic setting, every Bell outcome is accepted and the unconditional output is obtained by averaging over $m$,
\begin{eqnarray}
\hat{\rho}_{\mathrm{out}}=\cE(\hat{\rho}_{\mathrm{in}}),  \quad  \cE:=\int_{\C} d^{2}m \, \cE_{m},
\label{eq:deterministic_map}
\end{eqnarray}
so that $\cE$ is completely positive and trace preserving (CPTP). For the canonical BK implementation with a two-mode squeezed vacuum (TMSV) resource and unity-gain feed-forward, $\cE$ is a phase-insensitive Gaussian channel. In particular, in the idealized lossless model it can be written in the additive-noise form~\cite{BraunsteinKimble1998,BraunsteinvanLoock2005,Weedbrook2012}
\begin{eqnarray}
\cE_{\nu}(\hat{\rho}) = \int_{\C}d^{2}\xi \, \frac{1}{\pi \nu} \exp\left(-\frac{\abs{\xi}^{2}}{\nu}\right)\hD(\xi)\hat{\rho}\hD(\xi)^{\dagger},
\label{eq:additive_noise}
\end{eqnarray}
where $\nu \ge 0$ quantifies the added noise (for an ideal BK channel with squeezing parameter $r$, one has $\nu=e^{-2r}$ in shot-noise units). Eq.~(\ref{eq:additive_noise}) makes explicit that the deterministic BK teleportation is a random displacement of the input with an isotropic Gaussian kernel.

Because coherent states are displaced vacuum states, Eq.~(\ref{eq:additive_noise}) immediately yields
\begin{eqnarray}
\cE_{\nu}\left(\ketbra{\alpha}{\alpha}\right) \! = \! \int_{\C}d^{2}\xi \, \frac{1}{\pi \nu}\exp\!\left(-\frac{\abs{\xi}^{2}}{\nu}\right)\ketbra{\alpha+\xi}{\alpha+\xi},
\end{eqnarray}
i.e., the output is a Gaussian mixture of the coherent states centered at $\alpha$. The corresponding fidelity is then independent of $\alpha$,
\begin{eqnarray}
f(\alpha) \! = \! \int_{\C}d^{2}\xi \, \frac{1}{\pi \nu}\exp\!\left(-\frac{\abs{\xi}^{2}}{\nu}\right)\abs{\braket{\alpha}{\alpha+\xi}}^{2} \! = \! \frac{1}{1+\nu},
\label{eq:f_additive_const}
\end{eqnarray}
which illustrates the displacement-covariant baseline $D=0$ for deterministic unity-gain teleportation.

{\em Conditional (post-selected) teleportation.}---Many non-Gaussian enhancements are heralded and therefore probabilistic, including photon-subtraction-based resource engineering~\cite{Opatrny2000,DellAnno2007} and measurement-based noiseless linear amplification (MB-NLA) acting on the Bell record~\cite{Chrzanowski2014,Zhao2023,Fiurasek2024}. A convenient abstraction, which will be central for our later analysis, is to introduce an acceptance function $w(m) \in [0,1]$ that specifies how the experiment post-selects the measurement outcomes. The corresponding success map is then
\begin{eqnarray}
\cE_{\mathrm{succ}}(\hat{\rho}_{\mathrm{in}}) := \int_{\C} d^{2}m \, w(m) \cE_{m}(\hat{\rho}_{\mathrm{in}}),
\label{eq:succ_map_def}
\end{eqnarray}
with success probability $P_{\mathrm{succ}}(\hat{\rho}_{\mathrm{in}})=\Tr[\cE_{\mathrm{succ}}(\hat{\rho}_{\mathrm{in}})]$. The conditional output state on successful runs is
\begin{eqnarray}
\hat{\rho}_{\mathrm{out}}^{(\mathrm{succ})}(\hat{\rho}_{\mathrm{in}}) := \frac{\cE_{\mathrm{succ}}(\hat{\rho}_{\mathrm{in}})}{P_{\mathrm{succ}}(\hat{\rho}_{\mathrm{in}})}.
\label{eq:cond_output_def}
\end{eqnarray}
The deterministic protocol is recovered by setting $w(m) \equiv 1$, in which case $\cE_{\mathrm{succ}}=\cE$ and $P_{\mathrm{succ}}(\hat{\rho}_{\mathrm{in}})=1$ for all inputs. When $w(m)$ is nontrivial, however, the accepted ensemble is biased toward particular Bell outcomes, and this bias typically correlates with the input displacement. This is the structural origin of the increased input dependence (large fidelity deviation) in many conditional non-Gaussian schemes.

\subsection{Single-shot fidelity, average fidelity, and fidelity deviation}\label{subsec:figures}

For a pure target state $\ket{\psi}$ and an output state $\hat{\rho}$, the fidelity reduces to the overlap $\bra{\psi}\hat{\rho}\ket{\psi}$. Accordingly, for the coherent-state teleportation, we define the single-shot fidelity profile as
\begin{eqnarray}
f(\alpha) := \bra{\alpha} \cE\left(\ketbra{\alpha}{\alpha}\right) \ket{\alpha}.
\label{eq:single_fidelity}
\end{eqnarray}
To quantify the performance over a family of inputs, we work with an explicitly specified input ensemble and define the average fidelity $F$ and the fidelity deviation $D$ as the first two moments of $f(\alpha)$,
\begin{eqnarray}
F &:=& \int_{\C}d^{2}\alpha \, p_{\sigma}(\alpha) f(\alpha), \nonumber \\
D &:=& \sqrt{\int_{\C}d^{2}\alpha \, p_{\sigma}(\alpha) f(\alpha)^{2}-F^{2}}.
\label{eq:FD_det}
\end{eqnarray}
The quantity $F$ captures ``how well on average'' the protocol performs on the prescribed ensemble, whereas $D$ captures ``how uniformly'' it performs across that ensemble~\cite{BangRyuKaszlikowski2018,Patra2022}. In particular, $D=0$ holds if and only if the fidelity profile is constant on the support of $p_{\sigma}$, while large $D$ signals strong input selectivity.

In probabilistic protocols, the operational figure of merit is the conditional fidelity on successful runs. For a coherent input $\ket{\alpha}$, we introduce the success probability
\begin{eqnarray}
P_{\mathrm{succ}}(\alpha) := \Tr\left[\cE_{\mathrm{succ}}\left(\ketbra{\alpha}{\alpha}\right)\right],
\end{eqnarray}
and the conditional fidelity
\begin{eqnarray}
f_{\mathrm{succ}}(\alpha) &:=& \bra{\alpha}\hat{\rho}_{\mathrm{out}}^{(\mathrm{succ})}(\ketbra{\alpha}{\alpha})\ket{\alpha} \nonumber \\
	&=& \frac{1}{P_{\mathrm{succ}}(\alpha)}\bra{\alpha}\cE_{\mathrm{succ}}\left(\ketbra{\alpha}{\alpha}\right)\ket{\alpha}.
\label{eq:single_fidelity_cond}
\end{eqnarray}
Because the protocol is evaluated only on accepted trials, the relevant input distribution is reweighted by $P_{\mathrm{succ}}(\alpha)$. We therefore define the ensemble success probability and the conditional moments as
\begin{eqnarray}
&& F := \frac{1}{P_{\mathrm{succ}}}\int_{\C}d^{2}\alpha \, p_{\sigma}(\alpha) P_{\mathrm{succ}}(\alpha) f_{\mathrm{succ}}(\alpha), \nonumber \\
&& D := \sqrt{\frac{1}{P_{\mathrm{succ}}}\int_{\C}d^{2}\alpha \, p_{\sigma}(\alpha) P_{\mathrm{succ}}(\alpha) f_{\mathrm{succ}}(\alpha)^{2}-F^{2}}, \nonumber \\
&& P_{\mathrm{succ}} := \int_{\C}d^{2}\alpha \, p_{\sigma}(\alpha) P_{\mathrm{succ}}(\alpha).
\label{eq:PDF_condi}
\end{eqnarray}

\subsection{Input ensemble and regularization}\label{subsec:prior}

Unlike the finite-dimensional teleportation, CV teleportation does not admit a uniform average over ``all pure inputs'' because the Hilbert space is infinite dimensional and the natural phase-space volume is non-normalizable. Any meaningful performance assessment must therefore specify an input ensemble that regularizes the energy and reflects the intended use case~\cite{Patra2022}. Following common practice in CV benchmarking~\cite{Hammerer2005,Patra2022,Fiurasek2024}, we focus on coherent states with a Gaussian prior over displacements,
\begin{eqnarray}
p_{\sigma}(\alpha)=\frac{1}{\pi\sigma^{2}}\exp\left(-\frac{|\alpha|^{2}}{\sigma^{2}}\right),
\label{eq:prior}
\end{eqnarray}
where $\sigma^{2}=\int d^{2}\alpha p_{\sigma}(\alpha) \abs{\alpha}^{2}$ sets the typical input energy. This ensemble isolates the role of phase-space displacements while keeping the analysis analytically tractable and experimentally relevant; it is also the sharpest setting for diagnosing the loss of universality, because the deterministic unity-gain BK teleportation is displacement covariant and therefore produces an $\alpha$-independent fidelity profile.

Because $p_{\sigma}(\alpha)$ is phase invariant, any phase-insensitive quantity $g(\alpha)=g(\abs{\alpha})$ can be averaged by a one-dimensional radial integral. Writing $\alpha=re^{i\phi}$ with $r \ge 0$ and using $d^{2}\alpha=r dr d\phi$, one finds
\begin{eqnarray}
\int_{\C} d^{2}\alpha \, p_{\sigma}(\alpha) g(\abs{\alpha}) = \int_{0}^{\infty}dr \, \frac{2r}{\sigma^{2}}e^{-r^{2}/\sigma^{2}} g(r).
\label{eq:radial_reduction}
\end{eqnarray}
This reduction will be repeatedly used when we compute $(F,D)$ for concrete teleportation models. In particular, it provides a clean baseline: for deterministic displacement-covariant teleportation channels, $f(\alpha)$ is constant and hence $D=0$, whereas conditional filters $w(m)$ generically induce an $\alpha$-dependent fidelity profile, and hence $D>0$.

\section{Conceptual core: non-Gaussianity versus selectivity}\label{sec:core}

\subsection{Why ``non-Gaussianity'' is not the right explanatory variable}\label{subsec:ng_problem}

Non-Gaussian resources and operations are frequently advertised as ``the'' ingredient that unlocks the performance beyond Gaussian limits in continuous-variable (CV) information processing. This intuition is broadly correct at the level of feasibility---for instance, several no-go statements for entanglement distillation and error correction hold under Gaussian operations and Gaussian resources---but it becomes slippery once one asks a more input-resolved question. In our setting, the relevant object is not merely whether the resource state or the circuit is Gaussian, but how the implemented teleportation map weights different regions of the input phase space. The fidelity deviation $D$ introduced in Sec.~\ref{sec:setting} is precisely designed to be sensitive to such input dependence.

A first conceptual obstacle is that ``non-Gaussianity'' itself is not a single operational variable. One may speak of non-Gaussianity of (i) the shared resource state, (ii) the local operations in the teleportation circuit, or (iii) the induced effective channel conditioned on a heralding event. These notions are inequivalent and can even move in opposite directions under the same experimental tuning. More concretely, Patra \emph{et al.} systematically evaluated both the average fidelity and the fidelity deviation for a variety of non-Gaussian resources and input ensembles, and found parameter regimes where de-Gaussification improves the mean fidelity while the deviation does not increase, and can even decrease~\cite{Patra2022}. This already rules out any universal implication of the form ``non-Gaussianity $\Rightarrow D\uparrow$''.

The deeper reason is structural. The standard unity-gain Braunstein--Kimble (BK) teleportation circuit implements a map that is covariant under phase-space displacements when it is run deterministically (i.e., when all Bell outcomes are accepted)~\cite{BraunsteinKimble1998}. For an arbitrary deterministic displacement-covariant channel $\cE$,
\begin{eqnarray}
\cE\left(\hD(\alpha)\hat{\rho}\hD(\alpha)^{\dagger}\right)=\hD(\alpha)\cE(\hat{\rho})\hD(\alpha)^{\dagger},
\label{eq:disp_cov_core}
\end{eqnarray}
the coherent-state fidelity profile is necessarily flat:
\begin{eqnarray}
f(\alpha) = \bra{\alpha}\cE\left(\ketbra{\alpha}{\alpha}\right)\ket{\alpha} = \bra{0}\cE\left(\ketbra{0}{0}\right)\ket{0},
\label{eq:f_alpha_flat_core}
\end{eqnarray}
where we used $\ket{\alpha}=\hD(\alpha)\ket{0}$ and Eq.~(\ref{eq:disp_cov_core}). Hence, for the coherent-state prior in Eq.~(\ref{eq:prior}), any deterministic displacement-covariant teleportation has $D=0$ by definition, regardless of whether the underlying shared resource is Gaussian. In this sense, a deterministic protocol can be highly non-Gaussian at the resource level, yet perfectly ``non-selective'' on the coherent-state inputs.

This observation motivates a shift of emphasis. What changes $D$ in practice is not ``non-Gaussianity'' per se, but the mechanism by which a protocol breaks the displacement covariance and thereby makes the performance input-dependent. The most prominent route is conditioning: photon subtraction and related de-Gaussification steps on the resource~\cite{Opatrny2000,DellAnno2007}, as well as measurement-based noiseless linear amplification (MB-NLA) filters on the Bell record~\cite{Chrzanowski2014,Zhao2023,Fiurasek2024}, are implemented as heralded instruments rather than deterministic channels. In such schemes, the success probability $P_{\mathrm{succ}}(\alpha)$ and the conditional fidelity $f_{\mathrm{succ}}(\alpha)$ defined in Eq.~(\ref{eq:single_fidelity_cond}) and Eq.~(\ref{eq:PDF_condi}) become nontrivial functions of the input displacement. The conditional improvement is then typically achieved by accepting only a subset of measurement outcomes, which implicitly privileges some region of the input ensemble over the rest. We refer to this input-dependent bias as \emph{selectivity}, and we treat it as the operational mediator between conditional non-Gaussian enhancement and fidelity deviation.

\subsection{Selectivity indices and their relation to $D$}\label{subsec:selectivity_indices}

The definitions in Sec.~\ref{sec:setting} already contain a natural selectivity measure. For a deterministic protocol, the pair $(F,D)$ in Eq.~(\ref{eq:FD_det}) can be read as the mean and the dispersion of the fidelity landscape $f(\alpha)$ under the prior $p_{\sigma}(\alpha)$. In this language, $D$ quantifies how strongly the protocol ``prefers'' certain inputs over others, and $D=0$ is equivalent to perfect input-uniformity on the chosen ensemble.

For probabilistic protocols, the selectivity has two intertwined sources.

\begin{remark}[Two sources of selectivity in probabilistic teleportation]\label{rem:two_sources}
In a heralded protocol, input non-uniformity can arise (i) through the success-probability reweighting that defines the effective prior $p_{\sigma}^{(\mathrm{succ})}(\alpha)$, and (ii) through residual input dependence of the conditional fidelity profile $f_{\mathrm{succ}}(\alpha)$ under that effective prior. The fidelity deviation $D$ (equivalently the selectivity index $S$) captures their combined effect as the dispersion of $f_{\mathrm{succ}}(\alpha)$ with respect to $p_{\sigma}^{(\mathrm{succ})}$.
\end{remark}
 The success weighting induces an \emph{effective} input distribution for the successful events,
\begin{eqnarray}
p_{\sigma}^{(\mathrm{succ})}(\alpha) := \frac{p_{\sigma}(\alpha)P_{\mathrm{succ}}(\alpha)}{P_{\mathrm{succ}}},
\label{eq:effective_prior_core}
\end{eqnarray}
with $P_{\mathrm{succ}}$ in Eq.~(\ref{eq:PDF_condi}). Under this distribution, the conditional figures of merit take the standard mean--variance form
\begin{eqnarray}
&& F = \int_{\C}d^{2}\alpha \, p_{\sigma}^{(\mathrm{succ})}(\alpha) f_{\mathrm{succ}}(\alpha), \nonumber \\
&& D^{2} = \Var_{p_{\sigma}^{(\mathrm{succ})}}\bigl(f_{\mathrm{succ}}(\alpha)\bigr).
\label{eq:D_as_variance_core}
\end{eqnarray}
This motivates the following definition.
\begin{definition}[Selectivity index]\label{def:selectivity}
For a teleportation protocol evaluated on the coherent-state prior $p_{\sigma}$, we define the ``selectivity index'' as the dispersion of the (conditional) fidelity profile over the effective ensemble of successful events:
\begin{eqnarray}
S := \sqrt{\Var_{p_{\sigma}^{(\mathrm{succ})}}\big(f_{\mathrm{succ}}(\alpha)\big)}.
\label{eq:selectivity_S}
\end{eqnarray}
By Eq.~(\ref{eq:D_as_variance_core}), one has $S=D$ for probabilistic protocols, and $p_{\sigma}^{(\mathrm{succ})}=p_{\sigma}$ with $f_{\mathrm{succ}}=f$ for deterministic protocols.
\end{definition}

Interpreting $D$ as the selectivity does not change the mathematics, but it changes what we regard as the explanatory variable. The point is that the non-Gaussianity is a property of a state or an operation, whereas $S$ is an operational property of the input-response function of the entire protocol. In particular, a conditional non-Gaussian enhancement can increase $F$ by manipulating $f_{\mathrm{succ}}(\alpha)$ and/or $P_{\mathrm{succ}}(\alpha)$ so that the effective ensemble $p_{\sigma}^{(\mathrm{succ})}$ concentrates on a region where the protocol is accurate. The same manipulation generically makes the response less uniform, i.e., $S$ larger. The central thesis of this paper is therefore the causal chain
\begin{eqnarray}
&& \text{non-Gaussian conditioning} \;\Longrightarrow\; \text{selectivity ($S$)} \nonumber \\
&& \qquad \;\Longrightarrow\; \text{fidelity deviation ($D$)},
\label{eq:ng_to_selectivity_chain}
\end{eqnarray}
which can be investigated without committing to any particular intrinsic measure of non-Gaussianity.

While $S(=D)$ is a global dispersion, it is sometimes useful to probe the selectivity through alternative, more ``differential'' indicators of input sensitivity. When the protocol is phase insensitive so that $f_{\mathrm{succ}}(\alpha)=f_{\mathrm{succ}}(r)$ with $r:=\abs{\alpha}$, we consider two examples:
\begin{eqnarray}
S_{1} &:=& \int_{\C}d^{2}\alpha \, p_{\sigma}^{(\mathrm{succ})}(\alpha)\abs{\frac{\partial}{\partial r}f_{\mathrm{succ}}(r)}_{r=\abs{\alpha}}, \nonumber \\
S_{2} &:=& \Var_{p_{\sigma}^{(\mathrm{succ})}}\big(\log f_{\mathrm{succ}}(\alpha)\big).
\label{eq:S_def}
\end{eqnarray}
The index $S_{1}$ captures the typical radial slope of the fidelity landscape over the successful ensemble, and is therefore sensitive to how sharply the protocol ``turns off'' as the input amplitude grows. The index $S_{2}$ instead measures multiplicative variability: because $\log f$ turns relative changes into additive ones, $S_{2}$ highlights regimes where the protocol strongly amplifies or suppresses fidelity on a logarithmic scale. 

\section{Main theorems: covariance, selectivity, and local trade-offs}\label{sec:theorems}

The goal of this section is to turn the narrative of Sec.~\ref{sec:core} into the statements that are (i) model-independent and (ii) sharp enough to isolate what is genuinely non-trivial in our setting, namely the logical chain
\begin{eqnarray}
&& \text{(conditional improvement)} \Longrightarrow \text{(covariance breaking)} \nonumber \\
&& \qquad \Longrightarrow \text{(selectivity)} \Longrightarrow \text{(nonzero deviation)}.
\end{eqnarray}
The results below do not presume Gaussianity of the operations or the resource; rather, Gaussianity enters later only when we instantiate the abstract objects by the BK teleporter and its measurement-based non-Gaussian variants in Sec.~\ref{sec:mbnla}. Throughout, the input ensemble is the coherent-state Gaussian prior of Sec.~\ref{subsec:prior},
\begin{eqnarray}
\mathbb{E}_{\sigma}[g(\alpha)] &:=& \int_{\C} d^{2}\alpha \, p_{\sigma}(\alpha) g(\alpha), \nonumber \\
p_{\sigma}(\alpha) &=& \frac{1}{\pi\sigma^{2}}e^{-|\alpha|^{2}/\sigma^{2}}.
\label{eq:Esigma_def}
\end{eqnarray}
When the protocol is probabilistic, we additionally keep track of the heralding probability $P_{\mathrm{succ}}$ (Sec.~\ref{subsec:figures}). In that case, $F$ and $D$ are understood as the figures of merit conditioned on success in the operational sense.

\subsection{Covariance pins down universality: why deterministic BK yields $D=0$}\label{subsec:thm_covariance}

A deterministic unity-gain BK teleporter is displacement-covariant, hence it does not ``see'' the coherent-state label $\alpha$ as a distinguishable classical parameter. It treats the orbit $\{\ket{\alpha}\}_{\alpha \in \C}$ uniformly. The following theorem makes this intuition precise and simultaneously identifies the single structural mechanism that can produce a nonzero deviation in our benchmark ensemble.
\begin{theorem}[Covariance implies vanishing deviation on a group orbit]
\label{thm:covariance_implies_D0}
Let $G$ be a group with a unitary representation $g \mapsto \hat{U}_g$ on a Hilbert space $\mathcal{H}$. Fix a reference pure state $\ket{\psi_{0}} \in \mathcal{H}$ and consider its orbit
\begin{eqnarray}
\ket{\psi_{g}} := \hat{U}_g \ket{\psi_{0}}.
\end{eqnarray}
Let $\cE$ be a quantum channel (CPTP map) that is $G$-covariant in the sense that
\begin{eqnarray}
\cE\left(\hat{U}_g \hat{\rho} \hat{U}_g^{\dagger}\right) = \hat{U}_g \cE(\hat{\rho}) \hat{U}_g^{\dagger} \quad (\forall g \in G, \ \forall \hat{\rho}).
\label{eq:G_covariance}
\end{eqnarray}
Define the single-shot fidelity on the orbit by
\begin{eqnarray}
f(g) := \bra{\psi_{g}} \cE(\ketbra{\psi_{g}}{\psi_{g}}) \ket{\psi_{g}}.
\label{eq:f_on_orbit}
\end{eqnarray}
Then, $f(g)$ is constant on the orbit, i.e.,
\begin{eqnarray}
f(g) = \bra{\psi_{0}} \cE(\ketbra{\psi_{0}}{\psi_{0}}) \ket{\psi_{0}} = f(e) \quad (\forall g \in G).
\label{eq:f_constant_orbit}
\end{eqnarray}
Consequently, for any probability measure $\mu$ supported on the orbit, the fidelity deviation defined with respect to $\mu$ vanishes: namely, $D=0$.
\end{theorem}

\begin{proof}---By covariance in Eq.~(\ref{eq:G_covariance}) applied to $\hat{\rho}_0=\ketbra{\psi_0}{\psi_0}$ we have
\begin{eqnarray}
\cE(\ketbra{\psi_{g}}{\psi_{g}}) = \cE\left(\hat{U}_g \hat{\rho}_0 \hat{U}_g^{\dagger}\right) = \hat{U}_g \cE(\hat{\rho}_0)\hat{U}_g^{\dagger}.
\end{eqnarray}
Substituting into Eq.~(\ref{eq:f_on_orbit}) and using $\ket{\psi_g}=\hat{U}_g\ket{\psi_0}$,
\begin{eqnarray}
f(g) &=& \bra{\psi_0}\hat{U}_g^{\dagger}\left(\hat{U}_g \cE(\hat{\rho}_0) \hat{U}_g^{\dagger}\right) \hat{U}_g\ket{\psi_0} \nonumber \\
	&=& \bra{\psi_0} \cE(\hat{\rho}_0) \ket{\psi_0} = f(e),
\end{eqnarray}
which proves Eq.~(\ref{eq:f_constant_orbit}). Since $f$ is constant on the support of $\mu$, its variance with respect to $\mu$ is zero; hence the fidelity deviation $D$ is equal to $0$.
\end{proof}

{\bf Theorem~\ref{thm:covariance_implies_D0}} is deliberately abstract: it does not rely on Gaussianity, squeezing, or any special structure beyond covariance. Its operational message is that the universality (i.e., $D=0$) is not a numerical accident but a symmetry statement: if the input ensemble is a group orbit and the effective teleportation channel respects the same group action, then the fidelity cannot depend on the orbit label.

\begin{corollary}[Displacement covariance $\Rightarrow$ $D=0$ for coherent-state ensembles]
\label{cor:displacement_covariance}
Let $G$ be the Weyl--Heisenberg displacement group with $\hat{U}_{\alpha}=\hat{D}(\alpha)$ and $|\psi_0\rangle=|0\rangle$, so that $|\psi_{\alpha}\rangle=|\alpha\rangle$. If a teleportation channel $\cE$ is displacement-covariant at unity gain,
\begin{eqnarray}
\cE\left(\hat{D}(\alpha)\hat{\rho}\hat{D}^{\dagger}(\alpha)\right) = \hat{D}(\alpha)\cE(\hat{\rho})\hat{D}^{\dagger}(\alpha),
\label{eq:displacement_covariance}
\end{eqnarray}
then the coherent-state fidelity $f(\alpha)=\bra{\alpha}\cE(\ketbra{\alpha}{\alpha})\ket{\alpha}$ is independent of $\alpha$, and therefore $D=0$ for any prior $p(\alpha)$, including the Gaussian prior in Eq.~(\ref{eq:Esigma_def}).
\end{corollary}

\begin{proof}---This is {\bf Theorem~\ref{thm:covariance_implies_D0}} with $\hat{U}_g=\hat{D}(\alpha)$ and the coherent-state orbit.
\end{proof}

\begin{corollary}[Non-Gaussian resources do not imply nonzero deviation]\label{cor:NG_resource_not_imply_D}
If a CV teleportation implementation is deterministic and displacement-covariant at unity gain, then for coherent-state benchmarking the fidelity profile is flat and the deviation vanishes ($D=0$), irrespective of whether the shared entangled resource state is Gaussian or non-Gaussian.
\end{corollary}

\begin{proof}---Immediate from {\bf Corollary~\ref{cor:displacement_covariance}}.
\end{proof}

In the BK setting, unity-gain deterministic teleportation is well known to reduce to an additive-noise channel that is displacement-covariant~\cite{BraunsteinKimble1998}. In our framework, this implies that for the coherent-Gaussian prior of Sec.~\ref{subsec:prior}, the deviation is pinned to $D=0$ at the deterministic baseline. Hence, any observation of $D>0$ under the same benchmarking convention must originate from a structural change that breaks covariance, such as a gain mismatch or, more importantly for this work, a conditional non-Gaussian filter acting on measurement outcomes~\cite{Fiurasek2024,Zhao2023}.

\subsection{Selective improvement forces dispersion: a local universality-cost bound}\label{subsec:thm_local_tradeoff}

{\bf Theorem~\ref{thm:covariance_implies_D0}} tells us when $D$ must vanish. The next question, which is the central quantitative point of this work, is how $D$ begins to grow once we depart from the covariant baseline in a direction that increases the average teleportation fidelity. Here, we formalize this as a small-parameter perturbation and show that, generically, any improvement in the average fidelity is accompanied by an increase in dispersion across the inputs, unless the improvement is perfectly uniform (which would require preserving covariance).

To avoid committing to a specific physical implementation at this stage, we phrase the statement directly in terms of the fidelity profile $f_{\theta}(\alpha)$, where $\theta$ is a parameter controlling the ``strength'' of a conditional non-Gaussian modification (e.g., filter gain or cut-off). The key assumption is that we recover the covariant baseline at $\theta=0$, hence $f_{\theta=0}(\alpha) \equiv f_0$ is constant.

\begin{theorem}[Local universality-cost relation]
\label{thm:local_universality_cost}
Let $p_{\sigma}(\alpha)$ be the Gaussian prior as in Eq.~(\ref{eq:Esigma_def}). Let $f_{\theta}(\alpha) \in [0,1]$ be a family of measurable fidelity profiles such that, for $\theta \to 0$,
\begin{eqnarray}
f_{\theta}(\alpha)= f_0 + \theta h(\alpha) + r_{\theta}(\alpha),
\label{eq:f_theta_expansion}
\end{eqnarray}
where $f_0$ is a constant, $h \in L^{2}(p_{\sigma})$, and the remainder satisfies $\| r_{\theta} \|_{L^{2}(p_{\sigma})}=o(\abs{\theta})$. Define the average fidelity and fidelity deviation (standard deviation) with respect to $p_{\sigma}$ by
\begin{eqnarray}
F(\theta) &:=& \mathbb{E}_{\sigma}[f_{\theta}(\alpha)], \nonumber \\
D(\theta) &:=& \sqrt{\mathbb{E}_{\sigma}[f_{\theta}(\alpha)^2]-F(\theta)^2}.
\label{eq:FD_theta_def}
\end{eqnarray}
Then, the following asymptotic expansions hold:
\begin{eqnarray}
F(\theta) &=& f_0 + \theta \mathbb{E}_{\sigma}[h(\alpha)] + o(\abs{\theta}),
\label{eq:F_linear} \\
D(\theta) &=& \abs{\theta} \sqrt{\Var_{\sigma}(h)} + o(\abs{\theta}),
\label{eq:D_linear}
\end{eqnarray}
where $\Var_{\sigma}(h):=\mathbb{E}_{\sigma}[h^2]-\mathbb{E}_{\sigma}[h]^2$. In particular, if $\mathbb{E}_{\sigma}[h]>0$ (mean fidelity improves to first order) and $h$ is not constant $p_{\sigma}$-almost surely (equivalently, $\Var_{\sigma}(h)>0$), then $D(\theta)>0$ for all sufficiently small $\theta \neq 0$.

Moreover, whenever $\mathbb{E}_{\sigma}[h]>0$, the ratio of deviation growth to mean improvement has a well-defined small-$\theta$ limit,
\begin{eqnarray}
\lim_{\theta\to 0^{+}} \frac{D(\theta)}{F(\theta)-f_0} = \frac{\sqrt{\Var_{\sigma}(h)}}{\mathbb{E}_{\sigma}[h]} =: c_{\mathrm{U}},
\label{eq:universality_cost_slope}
\end{eqnarray}
which we interpret as the ``local universality-cost coefficient''.
\end{theorem}

\begin{proof}---From Eq.~(\ref{eq:f_theta_expansion}) and the linearity of $\mathbb{E}_{\sigma}$,
\begin{eqnarray}
F(\theta)=\mathbb{E}_{\sigma}[f_{\theta}] &=& \mathbb{E}_{\sigma}\left[ f_0 + \theta h + r_{\theta} \right] \nonumber \\
	&=& f_0 + \theta\mathbb{E}_{\sigma}[h] + \mathbb{E}_{\sigma}[r_{\theta}].
\end{eqnarray}
By Cauchy--Schwarz, $\abs{\mathbb{E}_{\sigma}[r_{\theta}]} \le \| r_{\theta} \|_{L^{2}(p_{\sigma})} = o(\abs{\theta})$, proving Eq.~(\ref{eq:F_linear}).

For the second moment,
\begin{eqnarray}
\mathbb{E}_{\sigma}[f_{\theta}^2] &=& \mathbb{E}_{\sigma}\!\left[(f_0+\theta h+r_{\theta})^2\right] \nonumber \\
	&=& f_0^2 + 2\theta f_0 \mathbb{E}_{\sigma}[h] + \theta^{2}\mathbb{E}_{\sigma}[h^2] + 2f_0\mathbb{E}_{\sigma}[r_{\theta}] \nonumber \\
	&& \quad +2\theta\mathbb{E}_{\sigma}[h r_{\theta}] + \mathbb{E}_{\sigma}[r_{\theta}^2].
\label{eq:second_moment_expand}
\end{eqnarray}
The remainder terms satisfy
\begin{eqnarray}
&& \mathbb{E}_{\sigma}[r_{\theta}] = o(\abs{\theta}), \nonumber \\[2pt]
&& \abs{\mathbb{E}_{\sigma}[h\,r_{\theta}]} \le \| h \|_{2} \| r_{\theta} \|_{2} = o(\abs{\theta}), \nonumber \\[2pt]
&& \mathbb{E}_{\sigma}[r_{\theta}^2] = \| r_{\theta} \|_{2}^2 = o(\theta^2).
\end{eqnarray}
Hence, Eq.~(\ref{eq:second_moment_expand}) becomes
\begin{eqnarray}
\mathbb{E}_{\sigma}[f_{\theta}^2] = f_0^2 +2\theta f_0 \mathbb{E}_{\sigma}[h] + \theta^{2}\mathbb{E}_{\sigma}[h^2] + o(\theta^2).
\label{eq:second_moment_asymp}
\end{eqnarray}
Similarly, from Eq.~(\ref{eq:F_linear}),
\begin{eqnarray}
F(\theta)^2 = f_0^2 +2\theta f_0 \mathbb{E}_{\sigma}[h] + \theta^2\mathbb{E}_{\sigma}[h]^2 + o(\theta^2).
\end{eqnarray}
Subtracting this from Eq.~(\ref{eq:second_moment_asymp}) yields
\begin{eqnarray}
D(\theta)^2 = \theta^2 \Var_{\sigma}(h) + o(\theta^2),
\end{eqnarray}
and taking the square roots, we obtain Eq.~(\ref{eq:D_linear}). Finally, using Eq.~(\ref{eq:F_linear}) in the form $F(\theta)-f_{0}=\theta\mathbb{E}_{\sigma}[h]+o(\abs{\theta})$ and dividing Eq.~(\ref{eq:D_linear}) by this expression, we obtain Eq.~(\ref{eq:universality_cost_slope}) in the limit $\theta \to 0^{+}$.
\end{proof}

{\bf Theorem~\ref{thm:local_universality_cost}} clarifies why the claim ``non-Gaussianity increases $D$'' is not the right mathematical target: the coefficient governing the deviation growth is not a generic non-Gaussianity measure but the variance of the first-order fidelity response $h(\alpha)$ under the benchmark prior. In particular, if a modification improves the mean fidelity through a response $h(\alpha)$ that is sharply peaked in $\abs{\alpha}$ (high selectivity), then $\Var_{\sigma}(h)$ is large, which forces a large $D$ already at small $\theta$. Conversely, if an improvement is uniform ($h(\alpha) \equiv \mathrm{const.}$), then $\Var_{\sigma}(h)=0$ and the protocol can in principle move along the $D=0$ axis, but this can occur only if the underlying symmetry (covariance) is preserved. This is precisely the sense in which selectivity is the operational mediator between the conditional improvement and the universality loss.

For later use, note that since our selectivity index in Sec.~\ref{sec:core} is essentially the variance of $f(\alpha)$, {\bf Theorem~\ref{thm:local_universality_cost}} also gives, to leading order,
\begin{eqnarray}
S(\theta)^{2}=D(\theta)^{2} = \theta^{2} \Var_{\sigma}(h) + o(\theta^{2}),
\end{eqnarray}
which makes the following causal chain explicit at the level of an expansion:
\begin{eqnarray}
\text{(covariance breaking)} \Longrightarrow S>0 \Longrightarrow D>0.
\end{eqnarray}

\subsection{Including heralding probability: a quality--throughput relation}\label{subsec:thm_triple_tradeoff}

A conditional protocol has a second axis of cost that a deterministic protocol does not, that is, the heralding probability $P_{\mathrm{succ}}$. In applications, the physically relevant question is often not only ``how large is $F$ on success?'' but also ``how frequently do we obtain a state that is reliably close to the target?'' This motivates treating $(F, D, P_{\mathrm{succ}})$ as a triple of figures of merit. The next theorem provides a minimal model-independent link between the three.

\begin{theorem}[Guaranteed high-fidelity success rate from $(F, D, P_{\mathrm{succ}})$]
\label{thm:throughput_bound}
Consider any probabilistic teleportation protocol with overall success probability $P_{\mathrm{succ}}$ under the benchmark experiment of Sec.~\ref{subsec:figures}. Let $X \in [0,1]$ denote the random fidelity on a successful run, where the randomness is over the input drawn from $p_{\sigma}$ and the internal measurement outcomes conditioned on the success. Let $F=\mathbb{E}[X]$ and $D=\sqrt{\Var(X)}$ be the conditional mean and deviation.

Then, for any $\delta>0$, one has the concentration bound
\begin{eqnarray}
\Pr\left[X \ge F-\delta \mid \mathrm{succ}\right] \ge 1-\frac{D^2}{\delta^2}.
\label{eq:chebyshev_conditional}
\end{eqnarray}
Equivalently, the per-trial probability of producing an output state whose fidelity is at least $F-\delta$ is bounded as
\begin{eqnarray}
\Pr\left[\mathrm{succ} \wedge (X \ge F-\delta)\right] \ge P_{\mathrm{succ}} \left(1-\frac{D^2}{\delta^2}\right).
\label{eq:chebyshev_unconditional}
\end{eqnarray}
\end{theorem}

\begin{proof}---Eq.~(\ref{eq:chebyshev_conditional}) is the Chebyshev inequality applied to $X$ (conditioned on success):
\begin{eqnarray}
\Pr\left(\abs{X-F} \ge \delta \mid \mathrm{succ}\right) \le \frac{D^2}{\delta^2}.
\end{eqnarray}
Since $\{X<F-\delta\} \subseteq \{\abs{X-F} \ge \delta\}$, we obtain
\begin{eqnarray}
\Pr\left(X<F-\delta \mid \mathrm{succ}\right) \le \frac{D^2}{\delta^2},
\end{eqnarray}
which is equivalent to Eq.~(\ref{eq:chebyshev_conditional}). Multiplying both sides by $\Pr(\mathrm{succ})=P_{\mathrm{succ}}$ yields Eq.~(\ref{eq:chebyshev_unconditional}).
\end{proof}

{\bf Theorem~\ref{thm:throughput_bound}} turns the dispersion $D$ into a throughput-relevant quantity. Even when $F$ is high, a large $D$ means that a non-negligible fraction of successful events can still have fidelity far below $F$, so the rate of producing reliably good outputs can be much smaller than the success rate itself. In the next Sec.~\ref{sec:mbnla}, we will show that measurement-based non-Gaussian filters can indeed raise $F$, but the resulting points in $(F, D, P_{\mathrm{succ}})$ space lie on a constrained Pareto-like surface: improving $F$ without paying in either $D$ (selectivity) or $P_{\mathrm{succ}}$ (heralding rate) is generically impossible, and the mechanism is precisely the covariance breaking identified in {\bf Theorem~\ref{thm:covariance_implies_D0}}.

\section{Worked example: MB-NLA-enhanced teleportation}\label{sec:mbnla}

We now instantiate our ``selectivity'' mechanism with a concrete protocol of CV teleportation enhanced by measurement-based noiseless linear amplification (MB-NLA)~\cite{RalphLund2009,Xiang2010,Chrzanowski2014,Zhao2023,Fiurasek2024}. The key idea is to herald improved performance by filtering the continuous Bell-measurement record, so that the conditional teleportation map deviates from the displacement-covariant BK channel. In the language developed above, MB-NLA is an ideal testbed for the implication
\begin{eqnarray}
&& \text{covariance breaking} \nonumber \\
&& \quad ~\Longrightarrow~ \text{selective $f_{\mathrm{succ}}(\alpha)$} ~\Longrightarrow~ D>0,
\end{eqnarray}
as well as for the local universality-cost relation of {\bf Theorem~\ref{thm:local_universality_cost}}.

\subsection{Model formulation: BK teleportation with an MB-NLA filter}

In the BK protocol, the Bell-measurement produces a complex record $m \in \C$ (equivalently, a pair of real quadrature outcomes), and the protocol, under deterministic unity-gain feed-forward, yields a displacement-covariant effective channel. As proven in {\bf Theorem~\ref{thm:covariance_implies_D0}}, this implies a perfectly flat fidelity profile for the coherent-state benchmarking, hence $D=0$ for the coherent-state Gaussian prior used here.

MB-NLA modifies this by introducing a radial acceptance weight on the Bell record. Following the analyses in Refs.~\cite{Chrzanowski2014,Zhao2023,Fiurasek2024}, we take the bounded filter
\begin{eqnarray}
w_{g,m_{c}}(m) = \Theta(m_{c} - \abs{m}) e^{\left(1-\frac{1}{g^{2}}\right)\left(\abs{m}^{2} - m_{c}^{2}\right)},
\label{eq:mbnla_filter}
\end{eqnarray}
where $g>1$ is the nominal amplification gain, $m_{c}>0$ is a hard cut-off, and $\Theta$ is the Heaviside step function~\cite{Zhao2023,Fiurasek2024,Shajilal2024}. The prefactor $e^{-(1-1/g^{2})m_{c}^{2}}$ ensures $0 \leq w_{g,m_{c}}(m) \leq 1$ so that $w_{g,m_{c}}$ can be interpreted as an accept probability~\footnote{The inverse-Gaussian dependence inside the cut-off is the measurement-based analogue of the unbounded NLA map; the hard cut-off is necessary for physical implementability and is the structural origin of non-Gaussianity and covariance breaking.}.

A fully microscopic derivation of the resulting conditional teleportation channel can be carried out in phase space~\cite{Fiurasek2024}. Here, we adopt a minimal correlated-Gaussian surrogate that isolates the selectivity mechanism while remaining analytically transparent and numerically stable. We model the Bell record as a noisy estimate of the unknown coherent displacement,
\begin{eqnarray}
m=\alpha+n,
\label{eq:m_model}
\end{eqnarray}
where $n$ is circular complex Gaussian noise with variance $V_{n}$, $p_{n}(n)=\frac{1}{\pi V_{n}}e^{-|n|^{2}/V_{n}}$. We further model the residual displacement error on the teleported output as a sum of a component correlated with $n$ and an independent component,
\begin{eqnarray}
\delta=\kappa n + \epsilon,
\label{eq:delta_model}
\end{eqnarray}
where $\epsilon$ is also circular complex Gaussian with variance $V_{\epsilon}$ and independent of $n$. The parameter $\kappa \in \R$ captures how much of the output error is predictable from the Bell records. In particular, in the deterministic unity-gain BK limit, the output state after feed-forward is independent of $m$, which corresponds to $\kappa=0$ in the present surrogate; in that case, any filtering of $m$ is necessarily futile because it cannot affect the conditional output. By contrast, MB-NLA protocols deliberately introduce a controlled correlation between the feed-forward correction and the Bell record so that heralding can sharpen the effective displacement estimate~\cite{Chrzanowski2014,Fiurasek2024}. This is precisely the regime $\kappa \neq 0$.

For a coherent input $\ket{\alpha}$, Eq.~(\ref{eq:m_model}) and Eq.~(\ref{eq:delta_model}) induce the following (unconditioned) output mixture,
\begin{eqnarray}
\hat{\rho}_{\mathrm{out}}(\alpha) = \int_{\C} d^{2}\delta \, p(\delta) \ketbra{\alpha+\delta}{\alpha+\delta},
\label{eq:mixture_model}
\end{eqnarray}
with $p(\delta)$ the complex Gaussian distribution implied by Eq.~(\ref{eq:delta_model}). Conditioning on the heralded success amounts to reweighting the $n$-distribution by the acceptance function $w_{g,m_{c}}(\alpha+n)$ and renormalizing, which produces a non-Gaussian effective map for any finite $m_{c}$.

\subsection{Single-shot conditional fidelity profile}

For a given input $\ket{\alpha}$, the success probability is
\begin{eqnarray}
P_{\mathrm{succ}}(\alpha) = \int_{\C} d^{2}n \, p_{n}(n) w_{g,m_{c}}(\alpha+n),
\label{eq:Psucc_alpha_mbnla}
\end{eqnarray}
and the conditional single-shot fidelity is $f_{\mathrm{succ}}(\alpha)=\bra{\alpha}\hat{\rho}_{\mathrm{out}}^{(\mathrm{succ})}(\alpha)\ket{\alpha}$ as defined in Eq.~(\ref{eq:single_fidelity_cond}).

The correlated-Gaussian surrogate yields a closed expression for $f_{\mathrm{succ}}(\alpha)$ as a ratio of two explicit integrals.
\begin{proposition}[Ratio formula for $f_{\mathrm{succ}}(\alpha)$]\label{prop:f_ratio_mbnla}
Under Eq.~(\ref{eq:m_model}), Eq.~(\ref{eq:delta_model}) and the MB-NLA filter in Eq.~(\ref{eq:mbnla_filter}), the conditional single-shot fidelity is
\begin{eqnarray}
f_{\mathrm{succ}}(\alpha) = \frac{\int_{\C} d^{2}n \, p_{n}(n) w_{g,m_{c}}(\alpha+n) e^{-\frac{\kappa^{2}}{1+V_{\epsilon}}\abs{n}^{2}}}{(1+V_{\epsilon}) \int_{\C} d^{2}n \, p_{n}(n) w_{g,m_{c}}(\alpha+n)}.
\label{eq:f_alpha_ratio}
\end{eqnarray}
In particular, when $w_{g,m_{c}} \equiv 1$ (deterministic teleportation without heralding),
\begin{eqnarray}
f_{\mathrm{succ}}(\alpha) = \frac{1}{1+V_{\epsilon}+\kappa^{2}V_{n}} \quad (\forall\alpha \in \C),
\label{eq:f0_surrogate}
\end{eqnarray}
so that $D=0$ for coherent-state benchmarking.
\end{proposition}

\begin{proof}---Conditioned on a fixed record noise $n$, Eq.~(\ref{eq:delta_model}) implies that the output is a Gaussian mixture over $\epsilon$,
\begin{eqnarray}
\hat{\rho}(\alpha | n) \! = \! \int_{\C}d^{2}\epsilon \frac{1}{\pi V_{\epsilon}} e^{-\frac{\abs{\epsilon}^{2}}{V_{\epsilon}}} \ketbra{\alpha \! + \! \kappa n \! + \! \epsilon}{\alpha \! + \! \kappa n \! + \! \epsilon}.
\end{eqnarray}
Using the overlap $\abs{\braket{\alpha}{\alpha+\beta}}^{2}=\exp(-\abs{\beta}^{2})$ and shifting $\epsilon\mapsto \epsilon-\kappa n$ gives
\begin{eqnarray}
\bra{\alpha}\hat{\rho}(\alpha \mid n)\ket{\alpha} &=& \int_{\C}d^{2}\epsilon \, \frac{1}{\pi V_{\epsilon}} e^{-\frac{\abs{\epsilon-\kappa n}^{2}}{V_{\epsilon}}}e^{-\abs{\epsilon}^2} \nonumber \\
	&=& \frac{1}{1+V_{\epsilon}} e^{-\frac{\kappa^{2}}{1+V_{\epsilon}}|n|^{2}},
\label{eq:coherent_overlap_gauss}
\end{eqnarray}
where the last step is a standard Gaussian integral in $\C$. Averaging Eq.~(\ref{eq:coherent_overlap_gauss}) over $n$ with the success reweighting $p_{n}(n)w_{g,m_{c}}(\alpha+n)$ and normalizing by $P_{\mathrm{succ}}(\alpha)$ in Eq.~(\ref{eq:Psucc_alpha_mbnla}), we obtain Eq.~(\ref{eq:f_alpha_ratio}). Eq.~(\ref{eq:f0_surrogate}) follows by setting $w_{g,m_{c}} \equiv 1$ and evaluating $\int d^{2}n \, p_{n}(n) e^{-\frac{\kappa^{2}}{1+V_{\epsilon}}|n|^{2}} = \big(1+\frac{\kappa^{2}}{1+V_{\epsilon}}V_{n}\big)^{-1}$.
\end{proof}

Eq.~(\ref{eq:f_alpha_ratio}) makes the selectivity mechanism explicit. The denominator is the (unnormalized) probability of heralding success, whereas the numerator adds an extra Gaussian penalty that suppresses large record noise $\abs{n}$. Thus, whenever $\kappa \neq 0$, the filter preferentially accepts events with smaller $|n|$ and hence reduces the correlated part of the residual displacement error. However, the hard cut-off in $w_{g,m_{c}}(\alpha+n)$ couples this preference to the unknown displacement $\alpha$, so the improvement is not uniform over the ensemble.

Because the prior and the filter are rotationally symmetric, the fidelity profile depends only on the radial coordinate.

\begin{corollary}[Radial reduction]\label{cor:radial_reduction_mbnla}
If $p_{n}(n)$ is circular Gaussian and $w_{g,m_{c}}(m)$ depends only on $\abs{m}$, then $f_{\mathrm{succ}}(\alpha)=f_{\mathrm{succ}}(\abs{\alpha})$ for all $\alpha \in \C$.
\end{corollary}

\begin{proof}---For any phase $\phi$, perform the change of variables $n \mapsto e^{i\phi}n$ in Eq.~(\ref{eq:f_alpha_ratio}). Since $p_{n}$ is circular and $w_{g,m_{c}}$ depends only on the magnitude, both the numerator and denominator are invariant under $\alpha \mapsto e^{i\phi}\alpha$, hence the ratio depends only on $\abs{\alpha}$.
\end{proof}

Most importantly, for our main narrative, a finite cut-off enforces selectivity in a quantitative way.

\begin{proposition}[Finite cut-off implies amplitude selectivity]\label{prop:cutoff_selectivity}
Assume $\kappa \neq 0$ and that the acceptance function satisfies $w(m)=0$ whenever $\abs{m} > m_{c}$ for some finite $m_{c}$. Then, for all $\abs{\alpha} \ge m_{c}$,
\begin{eqnarray}
f_{\mathrm{succ}}(\alpha) \le \frac{1}{1+V_{\epsilon}} e^{-\frac{\kappa^{2}}{1+V_{\epsilon}}(\abs{\alpha}-m_{c})^{2}}.
\label{eq:tail_bound_falpha}
\end{eqnarray}
Consequently, for any input ensemble with nonzero weight outside the disk $\abs{\alpha} \le m_{c}$ (in particular, for the Gaussian prior $p_{\sigma}$), the fidelity deviation satisfies $D>0$.
\end{proposition}

\begin{proof}---On the support of $w(\alpha+n)$ for $\abs{\alpha} \ge m_{c}$, we necessarily have $\abs{\alpha+n} \le m_{c}$. By the triangle inequality, $\abs{n}=\abs{\alpha-(\alpha+n)}\ge \abs{\alpha} - \abs{\alpha+n} \ge \abs{\alpha} - m_{c}$, so $e^{-\frac{\kappa^{2}}{1+V_{\epsilon}}\abs{n}^{2}} \le e^{-\frac{\kappa^{2}}{1+V_{\epsilon}}(\abs{\alpha}-m_{c})^{2}}$. Applying this point-wise upper bound to the numerator of Eq.~(\ref{eq:f_alpha_ratio}), we get
\begin{eqnarray}
f_{\mathrm{succ}}(\alpha) \le \frac{e^{-\frac{\kappa^{2}}{1+V_{\epsilon}}(|\alpha|-m_{c})^{2}}}{1+V_{\epsilon}} \frac{\int d^{2}n \, p_{n}(n) w(\alpha+n)}{\int d^{2}n \, p_{n}(n) w(\alpha+n)},
\end{eqnarray}
which gives Eq.~(\ref{eq:tail_bound_falpha}). If the ensemble assigns nonzero probability to both a region where $f_{\mathrm{succ}}(\alpha)$ is close to its small-$\abs{\alpha}$ value and a region where Eq.~(\ref{eq:tail_bound_falpha}) enforces a strictly smaller value, then $f_{\mathrm{succ}}(\alpha)$ is not almost surely constant under the success-weighted measure, and therefore its variance is strictly positive, i.e., $D > 0$.
\end{proof}

{\bf Proposition~\ref{prop:cutoff_selectivity}} formalizes a statement often made qualitatively in the MB-NLA literature: a finite cut-off can enhance the fidelity only on a restricted amplitude range~\cite{Zhao2023,Fiurasek2024}. In our framework, this is precisely the emergence of selectivity and the consequent $D$ increase.

\subsection{\texorpdfstring{Trade-off diagrams in $(F,D,P_{\mathrm{succ}})$ and estimating the universality-cost coefficient}{Trade-off diagrams in (F,D,Psucc) and estimating the universality-cost coefficient}}

We now evaluate the conditional performance measures $(F, D, P_{\mathrm{succ}})$ defined in Eq.~(\ref{eq:PDF_condi}) by inserting the fidelity profile of Eq.~(\ref{eq:f_alpha_ratio}) and the success probability in Eq.~(\ref{eq:Psucc_alpha_mbnla}). Because of {\bf Corollary~\ref{cor:radial_reduction_mbnla}}, all ensemble integrals can be reduced to one-dimensional radial integrals (with the coherent-state Gaussian prior) without any loss of information. The numerical results are shown in Figs.~\ref{fig:falpha}--\ref{fig:FDPsucc_density}, where we fix the surrogate parameters $(V_{n},V_{\epsilon},\kappa,\sigma)=(0.5,0.1,0.6,2.0)$ and vary only the MB-NLA gain $g$ and cut-off $m_{c}$.

\begin{figure}[t]
\centering
\includegraphics[width=1.00\linewidth]{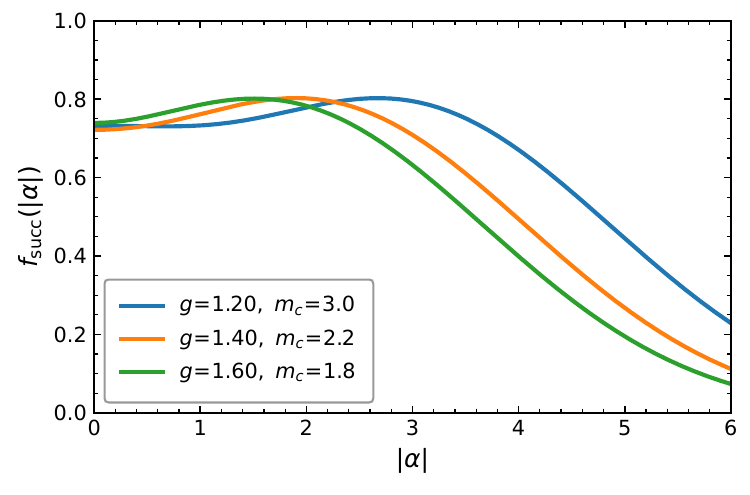}
\caption{Single-shot conditional fidelity profile $f_{\mathrm{succ}}(|\alpha|)$ under MB-NLA filtering for representative parameter choices.
A larger cut-off preserves a broader plateau and therefore keeps the protocol closer to uniform (smaller $D$), but it also leaves the scheme closer to the deterministic BK baseline.
A tighter cut-off can sharpen the near-origin performance, yet the improvement is necessarily localized in phase space and suppresses the large-displacement tail, as quantified by {\bf Proposition~\ref{prop:cutoff_selectivity}}.}
\label{fig:falpha}
\end{figure}

Fig.~\ref{fig:falpha} shows the single-shot conditional fidelity profile as a function of the input amplitude. The key feature is that MB-NLA does not lift the whole curve uniformly. Instead, stronger filtering redistributes the performance toward small and intermediate amplitudes. The mild setting $(g,m_{c})=(1.2, 3.0)$ keeps a broad, almost flat plateau and therefore remains relatively universal over the benchmark ensemble, whereas the more aggressive settings $(1.4, 2.2)$ and $(1.6, 1.8)$ sharpen the near-origin peak and force a faster decay at larger $\abs{\alpha}$. This is the selectivity mechanism already encoded in Eq.~(\ref{eq:f_alpha_ratio}): the heralding preferentially retains Bell records with smaller effective noise $\abs{n}$, but because the acceptance depends on $\abs{\alpha+n}$, the gain is concentrated in favored regions of phase space rather than distributed uniformly over all displacements. In this sense, Fig.~\ref{fig:falpha} is the most microscopic view of the universality cost: the enhanced fidelity near the origin is paid for by a compressed tail, and {\bf Proposition~\ref{prop:cutoff_selectivity}} guarantees that a finite cut-off ultimately turns this localization into a nonzero deviation.

\begin{figure}[t]
\centering
\includegraphics[width=1.00\linewidth]{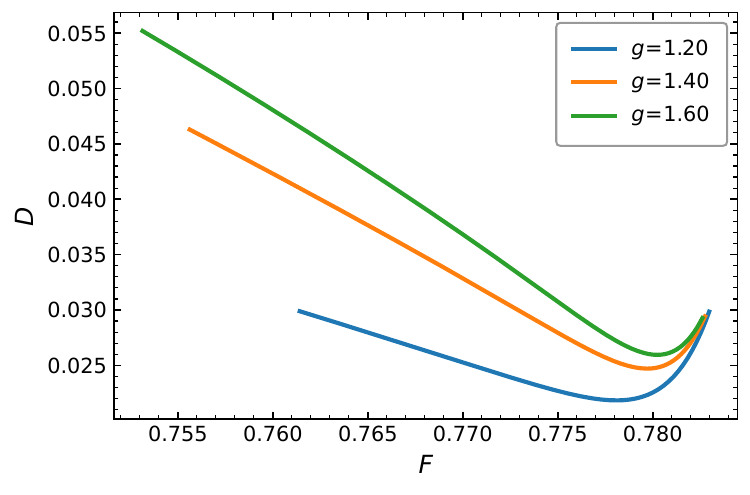}
\caption{Trade-off curves in the $(F,D)$ plane for MB-NLA-enhanced teleportation, obtained by sweeping the cut-off $m_{c}$ at fixed nominal gain $g$.
The rightmost tips are the weak-filter points closest to the deterministic baseline, while moving leftward and upward corresponds to stronger post-selection.
The slight upturn near the extreme right indicates that even the last incremental gain in $F$ can already be accompanied by an increase in $D$, signaling selectivity-driven improvement rather than a uniform channel enhancement.}
\label{fig:FD}
\end{figure}

Fig.~\ref{fig:FD} reorganizes the same data in the $(F, D)$ plane by sweeping $m_{c}$ at fixed $g$. The right-hand tip of each curve corresponds to weak filtering, where the protocol stays closest to the displacement-covariant BK baseline: $F$ is high and $D$ is comparatively small. To make this reference point explicit, note from Eq.~(\ref{eq:f0_surrogate}) that the no-heralding control $w \equiv 1$ gives a flat baseline
\begin{eqnarray}
f_{0}=\frac{1}{1+V_{\epsilon}+\kappa^{2}V_{n}}.
\end{eqnarray}
With the parameters used in Figs.~\ref{fig:falpha}--\ref{fig:FDPsucc_density}, namely $(V_{n},V_{\epsilon},\kappa,\sigma)=(0.5,0.1,0.6,2.0)$, this becomes
\begin{eqnarray}
f_{0}=\frac{1}{1+0.1+0.6^{2}\times 0.5}\approx 0.781.
\end{eqnarray}
The deterministic BK-type reference would therefore sit at $(F,D)=(f_{0},0)$, and the rightmost tips in Fig.~\ref{fig:FD} cluster around this value even though they still belong to fixed-$g$ heralded MB-NLA families and hence do not coincide exactly with the true deterministic baseline.

Moving leftward and upward along a curve tightens the post-selection, so the protocol becomes more selective and $D$ grows even when the average fidelity remains competitive. The separation between the three curves further shows that larger nominal gain permits stronger conditional enhancement but also makes the deviation rise more steeply. Near the high-$F$, low-$D$ branch, the trajectories are approximately linear, and the local slope is of order $c_{\mathrm{MB}} \sim 4$, which is consistent with the perturbative relation in Eq.~(\ref{eq:cMB_relation}) before the stronger-filter curvature sets in.

The small hook near the extreme right of each curve---where $F$ can increase slightly while $D$ also increases---is especially instructive. This behavior does not contradict {\bf Theorem~\ref{thm:local_universality_cost}}: the theorem is a local perturbative statement about the covariant baseline, whereas a finite sweep of $m_{c}$ within a fixed-$g$ heralded family can trace a mildly curved path in the $(F,D)$ plane once higher-order terms become visible. Operationally, the hook means that the last increment in average fidelity is already being purchased by extra selectivity. The successful events are concentrated a little more strongly on favored amplitudes, so the mean fidelity improves slightly, but the fidelity landscape becomes less uniform at the same time. In this sense, the upturn is not an anomaly but a direct visual signature of the universality cost discussed in Sec.~\ref{subsec:thm_local_tradeoff}.

\begin{figure}[t]
\centering
\includegraphics[width=1.00\linewidth]{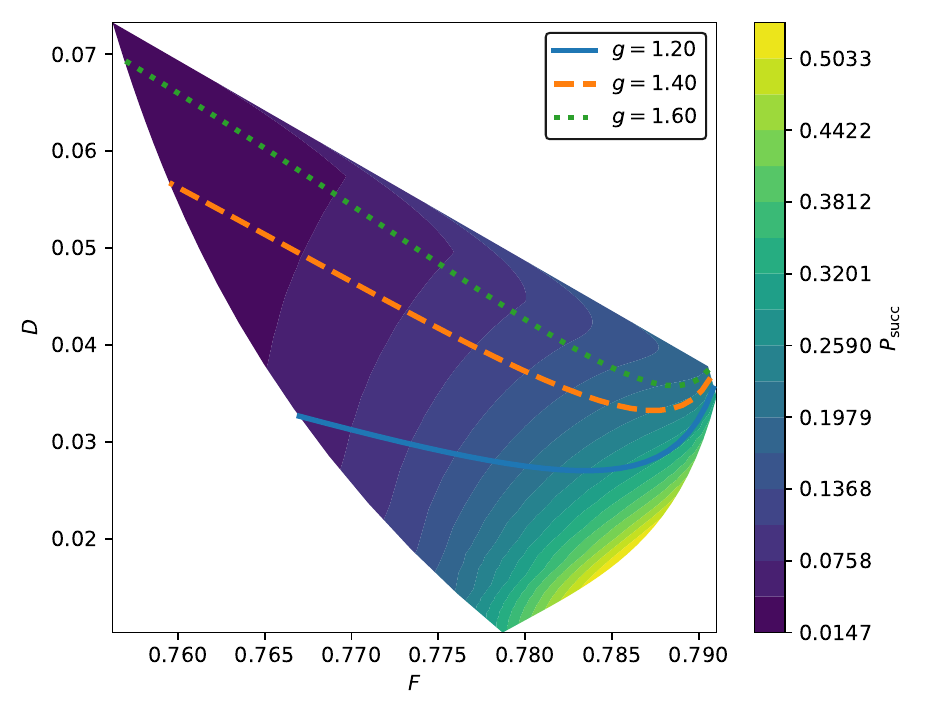}
\caption{Density plot of the joint trade-off in the $(F,D)$ plane for MB-NLA-enhanced teleportation, where color encodes $P_{\mathrm{succ}}$.
The bright lower-right region marks the operational sweet spot of simultaneously high fidelity, small deviation, and relatively high success probability.
Moving toward darker upper-right or upper-left regions reveals the cost of stronger filtering: the protocol can keep or even raise the conditional mean fidelity only by sacrificing universality, throughput, or both.}
\label{fig:FDPsucc_density}
\end{figure}

Fig.~\ref{fig:FDPsucc_density} is a density plot on the $(F, D)$ plane, with color encoding $P_{\mathrm{succ}}$. This presentation makes the three-way trade-off transparent. The lower-right region (high $F$, low $D$) is the operational sweet spot: the protocol is accurate on average, reasonably uniform, and the success probability is also highest. The upper-right region (high $F$, high $D$) is particularly important because it contains protocols that can still look attractive if one inspects only the mean fidelity, yet their good average performance is already achieved by favoring a narrower subset of inputs; the color scale shows that $P_{\mathrm{succ}}$ is already reduced there. The upper-left region (low $F$, high $D$) is the overfiltered regime, where both universality and throughput are poor and the success probability becomes smallest. Finally, the lower-left region (low $F$, low $D$) represents comparatively uniform but insufficiently corrected performance; here $P_{\mathrm{succ}}$ can be moderate, yet the fidelity gain is no longer substantial.

Read together with Fig.~\ref{fig:FD}, these trends show why the point of maximum $F$ is not automatically the best operating point. Once universality and throughput are accounted for, a point slightly away from the absolute maximum of $F$ can be preferable, especially if one optimizes a robustness-aware objective such as $J_{\lambda}=F-\lambda D$ or imposes explicit constraints such as $D \le D_{\max}$ and $P_{\mathrm{succ}} \ge P_{\min}$. The density plot therefore visualizes in a single panel the central practical lesson of this work: one cannot simultaneously maximize average fidelity, universality, and heralding rate, because pushing the filter harder inevitably moves the protocol away from the bright high-$F$/low-$D$ corner toward darker, more selective regions.

We now connect the worked example to the local universality-cost coefficient introduced in Sec.~\ref{subsec:thm_local_tradeoff}. In the weak-filter regime (large $m_{c}$ and $g \simeq 1$), the fidelity profile admits a small deformation $f_{\mathrm{succ}}(\alpha) = f_{0} + \theta h(\alpha) + O(\theta^{2})$ with a nontrivial selectivity profile $h$. {\bf Theorem~\ref{thm:local_universality_cost}} then predicts a linear relation
\begin{eqnarray}
D \simeq c_{\mathrm{MB}} [F-f_{0}],
\label{eq:cMB_relation}
\end{eqnarray}
where $c_{\mathrm{MB}}$ depends only on the first two moments of $h$ under the success-weighted input measure. Eq.~(\ref{eq:cMB_relation}) makes operational our central message: for MB-NLA-enhanced teleportation, a conditional fidelity increase is quantitatively accompanied by a predictable growth of input selectivity, unless the protocol remains effectively covariant on the support of the input ensemble.

\subsection{Resource non-Gaussianity vs filtered non-Gaussianity: where selectivity enters}\label{subsec:mbnla_vs_filter}

The control case above highlights that non-Gaussianity is not, by itself, an operational explanation of $D$. To reconcile this with the MB-NLA behavior in Sec.~\ref{sec:mbnla}, it is helpful to distinguish two qualitatively different ways in which non-Gaussianity appears in teleportation practice.

First, non-Gaussian resources are often prepared by a heralded operation (e.g., photon subtraction)~\cite{Walschaers2021,Thapliyal2024,Arora2025}, but the heralding occurs \emph{before} the unknown input is presented. Conditioned on successful preparation, the subsequent teleportation run is deterministic, and, crucially, the resource-preparation success probability is independent of the unknown displacement label $\alpha$. Such ``offline'' heralding therefore does not bias the input ensemble and does not introduce selectivity with respect to $\alpha$. Within the coherent-state benchmarking at unity gain, {\bf Corollary~\ref{cor:NG_resource_not_imply_D}} implies that the selectivity index remains $S=0$.

Second, MB-NLA is a filter on the Bell record that operates \emph{during} the teleportation and accepts the events depending on $m$, typically through a bounded, radially dependent weight $w(m)$ with an explicit cut-off~\cite{Chrzanowski2014,Zhao2023,Fiurasek2024}. Because the Bell record $m$ contains the information about the unknown displacement, such filtering correlates the success event with $\alpha$ and thereby reshapes the effective input distribution $p_{\sigma}^{(\mathrm{succ})}(\alpha)$ in Eq.~(\ref{eq:effective_prior_core}). This is precisely the origin of the selectivity in our framework. In MB-NLA, the improvement in $F$ is achieved in part by concentrating the successful events on a subset of inputs, and therefore a nonzero deviation is expected. Indeed, in Sec.~\ref{sec:theorems}, we proved that any nonuniform first-order improvement forces $D>0$ ({\bf Theorem~\ref{thm:local_universality_cost}}), and in Sec.~\ref{sec:mbnla}, we exhibited this mechanism explicitly through the amplitude-dependent conditional fidelity profile and the resulting trade-off geometry.

In summary, the control case validates the central claim of this paper: $D$ is not a generic marker of non-Gaussianity of the resource, but a marker of input selectivity induced by covariance breaking. Non-Gaussianity becomes relevant to $D$ precisely when it is deployed as an input-correlated filter (or, more generally, as a conditional transformation whose success probability depends on the measurement record that encodes the input label), which is the operative regime for MB-NLA and related post-selected enhancements.

\section{Non-Gaussianity--selectivity--universality: interpretation and design principles}\label{sec:design}

The results of Secs.~\ref{sec:theorems} and \ref{sec:mbnla} show that, for the coherent-state benchmarking, the fidelity deviation $D$ is fundamentally a symmetry diagnostic rather than a non-Gaussianity meter. The deterministic unity-gain BK teleportation is displacement covariant and therefore universal on the coherent-state orbit, forcing $D=0$ by {\bf Theorem~\ref{thm:covariance_implies_D0}} and {\bf Corollary~\ref{cor:displacement_covariance}}. By contrast, the post-selected enhancements, such as MB-NLA, alter the protocol into a filtering instrument, thereby correlating success with the Bell record and (through it) with the unknown displacement $\alpha$. 

This section provides two complementary layers of the interpretation. First, we give an information-theoretic view of why conditioning generically induces the input selectivity. Second, we translate the pair $(F, D)$ into design rules that explicitly balance mean performance against universality, thereby avoiding the pitfall of optimizing $F$ alone.

\subsection{Why conditional filters increase $D$: an information-theoretic view}\label{subsec:info_view}

In the conditional setting, the success event itself becomes a physical output of the device. Even if one ultimately cares only about the teleported quantum state, the binary record ``success/failure'' is observable, repeatable, and in practice often logged. Because the Bell record $m$ carries the information about the input displacement, any acceptance rule $w(m)$ that depends on $m$ makes the success flag statistically dependent on the unknown input label. In this sense, a conditional filter ``reveals more information'' about the input than a deterministic covariant teleporter. Thus, the success is no longer a neutral event but a Bayesian update about which inputs are more likely to have occurred~\cite{Chrzanowski2014,Fiurasek2024}.

To formalize this statement in a model-independent way, we treat the coherent amplitude $\alpha$ as a random variable distributed according to the prior $p_{\sigma}(\alpha)$, and we denote by $S \in \{0,1\}$ the success indicator of a heralded teleportation attempt.
By definition,
\begin{eqnarray}
&& \Pr(S=1 \mid \alpha) = P_{\mathrm{succ}}(\alpha), \nonumber \\
&& P_{\mathrm{succ}} := \Pr(S=1) = \int_{\C} d^2\alpha \, p_{\sigma}(\alpha) P_{\mathrm{succ}}(\alpha),
\label{eq:Ps_alpha_def_design}
\end{eqnarray}
where $P_{\mathrm{succ}}(\alpha)$ is the input-dependent success probability and $P_{\mathrm{succ}}$ is the ensemble success probability (cf.~Sec.~\ref{subsec:figures}). The following theorem makes precise in what sense $S$ necessarily carries information about the input whenever $P_{\mathrm{succ}}(\alpha)$ is not constant.

\begin{theorem}[Heralding reveals information about the input label]
\label{thm:mutual_info_success}
Let $\alpha \sim p_{\sigma}(\alpha)$ and $S \in \{0,1\}$ satisfy $\Pr(S=1 \mid \alpha)=P_{\mathrm{succ}}(\alpha)$ with $0 < P_{\mathrm{succ}} < 1$. Then, the mutual information between the input label and the success flag is
\begin{eqnarray}
I(\alpha;S) = h_{2}(P_{\mathrm{succ}}) - \int_{\C} d^2\alpha \, p_{\sigma}(\alpha) h_{2}\left(P_{\mathrm{succ}}(\alpha)\right),
\label{eq:I_alpha_S}
\end{eqnarray}
and it satisfies
\begin{eqnarray}
I(\alpha;S) \ge 0,
\end{eqnarray}
where $h_{2}(x):=-x\log x-(1-x)\log(1-x)$ is the binary entropy (logarithms are natural). Here, $I(\alpha;S)=0$ if and only if $P_{\mathrm{succ}}(\alpha)=P_{\mathrm{succ}}$ holds $p_{\sigma}$-almost surely.
\end{theorem}

\begin{proof}---By definition of the mutual information~\cite{CoverThomas2006},
\begin{eqnarray}
I(\alpha;S) = H(S) - H(S \mid \alpha),
\end{eqnarray}
where $H$ denotes Shannon entropy. Since $S$ is Bernoulli with parameter $P_{\mathrm{succ}}$, we have $H(S)=h_{2}(P_{\mathrm{succ}})$. Conditioned on a fixed $\alpha$, $S$ is Bernoulli with parameter $P_{\mathrm{succ}}(\alpha)$, hence
\begin{eqnarray}
H(S \mid \alpha) = \int d^2\alpha \, p_{\sigma}(\alpha) h_{2}(P_{\mathrm{succ}}(\alpha)),
\end{eqnarray}
which gives Eq.~(\ref{eq:I_alpha_S}). The nonnegativity follows from the concavity of $h_{2}$ and Jensen's inequality~\cite{CoverThomas2006}:
\begin{eqnarray}
\int d^2\alpha \, p_{\sigma}(\alpha) h_{2}(P_{\mathrm{succ}}(\alpha)) &\le& h_{2}\!\left(\int d^2\alpha \, p_{\sigma}(\alpha) P_{\mathrm{succ}}(\alpha)\right) \nonumber \\
	&=& h_{2}(P_{\mathrm{succ}}).
\end{eqnarray}
If $0<P_{\mathrm{succ}}<1$, $h_{2}$ is strictly concave on $(0,1)$, so equality holds if and only if $P_{\mathrm{succ}}(\alpha)$ is constant $p_{\sigma}$-almost surely.
\end{proof}
 
{\bf Theorem~\ref{thm:mutual_info_success}} shows that any input-dependent success probability implies $I(\alpha;S)>0$. In other words, observing success-and-failure updates one's knowledge of $\alpha$. This is a purely classical statement about the input label and the heralding record, yet it captures the operational core of why conditional filters tend to induce selectivity. Indeed, the effective input distribution on successful runs is the Bayesian posterior
\begin{eqnarray}
p_{\sigma}^{(\mathrm{succ})}(\alpha) = p_{\sigma}(\alpha)\,\frac{P_{\mathrm{succ}}(\alpha)}{P_{\mathrm{succ}}},
\label{eq:posterior_prior_design}
\end{eqnarray}
which is exactly the distribution used in the conditional averages defining $(F, D)$ in Eq.~(\ref{eq:PDF_condi}). Thus, post-selection does not merely change the channel, but also changes the ensemble with respect to which ``average performance'' is operationally evaluated.

It is often useful to quantify the strength of this ensemble distortion by the relative entropy between the posterior and the prior.

\begin{proposition}[A selectivity-information functional for heralding]
\label{prop:KL_selectivity}
Define
\begin{eqnarray}
I_{\mathrm{sel}} &:=& D_{\mathrm{KL}}\left(p_{\sigma}^{(\mathrm{succ})} \,\|\, p_{\sigma}\right) \nonumber \\
	&=& \int_{\C} d^2\alpha \, p_{\sigma}^{(\mathrm{succ})}(\alpha)\log\frac{p_{\sigma}^{(\mathrm{succ})}(\alpha)}{p_{\sigma}(\alpha)}.
\label{eq:Isel_def}
\end{eqnarray}
Then,
\begin{eqnarray}
I_{\mathrm{sel}} \!=\! \int_{\C} d^2\alpha \, p_{\sigma}^{(\mathrm{succ})}(\alpha) \log P_{\mathrm{succ}}(\alpha) \!-\! \log P_{\mathrm{succ}} \ge 0,
\label{eq:Isel_closed}
\end{eqnarray}
with equality if and only if $P_{\mathrm{succ}}(\alpha)=P_{\mathrm{succ}}$ holds $p_{\sigma}$-almost surely.
\end{proposition}

\begin{proof}---Substituting Eq.~(\ref{eq:posterior_prior_design}) into Eq.~(\ref{eq:Isel_def}), we obtain
\begin{eqnarray}
I_{\mathrm{sel}} &=& \int d^2\alpha \, p_{\sigma}^{(\mathrm{succ})}(\alpha) \log\left(\frac{P_{\mathrm{succ}}(\alpha)}{P_{\mathrm{succ}}}\right) \nonumber \\
	&=& \int d^2\alpha \, p_{\sigma}^{(\mathrm{succ})}(\alpha)\log P_{\mathrm{succ}}(\alpha) - \log P_{\mathrm{succ}},
\end{eqnarray}
which is Eq.~(\ref{eq:Isel_closed}). Nonnegativity and the equality condition are standard properties of the relative entropy, and here reduce to the condition that $p_{\sigma}^{(\mathrm{succ})}=p_{\sigma}$ almost surely, i.e., $P_{\mathrm{succ}}(\alpha)$ is constant.
\end{proof}

{\bf Proposition~\ref{prop:KL_selectivity}} supplies a compact diagnostic that is logically upstream of the deviation $D$. The functional $I_{\mathrm{sel}}$ captures the extent to which the success event concentrates the input ensemble (in an information-theoretic sense), whereas $D$ captures how unevenly the teleportation fidelity is distributed under that effective ensemble. In MB-NLA, the cut-off filter makes $P_{\mathrm{succ}}(\alpha)$ strongly amplitude dependent, hence $I_{\mathrm{sel}}$ becomes nonzero. When, in addition, the conditional fidelity profile $f_{\mathrm{succ}}(\alpha)$ is itself amplitude dependent (Sec.~\ref{sec:mbnla}), the variance typically grows. This is the precise sense in which ``more information revealed'' by conditioning translates into the universality loss: the act of selecting events based on a record correlated with $\alpha$ both distorts the input distribution and induces an $\alpha$-dependent fidelity landscape, thereby producing $D>0$ in the regimes that improve $F$.

\subsection{Design principles: robust objectives beyond mean fidelity}\label{subsec:robust_objectives}

The preceding discussion motivates a shift from optimizing mean performance alone to optimizing mean performance with a controlled loss of universality. Operationally, $D$ matters because the teleportation is often used as a subroutine, where the rare low-fidelity events can dominate downstream error budgets~\cite{Larsen2021,Hillmann2022}. A useful design goal is therefore to guarantee that ``most'' successful events exceed a desired fidelity threshold, rather than merely ensuring a high mean.

To formalize this idea, let $X \in [0,1]$ denote the random single-shot fidelity conditioned on the success, with
\begin{eqnarray}
\mathbb{E}[X]=F, \quad \sqrt{\Var(X)}=D,
\end{eqnarray}
where the randomness is over the input drawn from $p_{\sigma}^{(\mathrm{succ})}$. The following theorem shows that the linear combination $F - \lambda D$ is not an ad hoc scalarization: it is a certified lower-confidence bound on $X$.

\begin{theorem}[A robust fidelity guarantee from $(F,D)$]
\label{thm:Cantelli_Jlambda}
Let $X \in [0,1]$ be a random variable with mean $F$ and standard deviation $D$. Then, for any $\lambda > 0$,
\begin{eqnarray}
\Pr\left[X \ge F - \lambda D\right] \ge \frac{\lambda^{2}}{1+\lambda^{2}}.
\label{eq:Cantelli_bound_lambda}
\end{eqnarray}
Equivalently, for any target confidence level $\eta\in(0,1)$,
\begin{eqnarray}
\Pr\left[X \ge F - \sqrt{\frac{\eta}{1-\eta}} D \right] \ge \eta.
\label{eq:Cantelli_bound_eta}
\end{eqnarray}
\end{theorem}

\begin{proof}---We apply the one-sided Chebyshev (Cantelli) inequality to $X$~\cite{Wainwright2019}:
\begin{eqnarray}
\Pr\left[X-F \le - t\right] \le \frac{D^{2}}{D^{2}+t^{2}} \quad (t>0).
\label{eq:Cantelli_raw}
\end{eqnarray}
By setting $t=\lambda D$, we have
\begin{eqnarray}
\Pr\left[X < F - \lambda D\right] \le \frac{1}{1+\lambda^{2}},
\end{eqnarray}
which is equivalent to Eq.~(\ref{eq:Cantelli_bound_lambda}). Eq.~(\ref{eq:Cantelli_bound_eta}) follows by choosing $\lambda=\sqrt{\eta/(1-\eta)}$.
\end{proof}

{\bf Theorem~\ref{thm:Cantelli_Jlambda}} provides a direct interpretation of the robust objective
\begin{eqnarray}
J_{\lambda} := F - \lambda D.
\label{eq:Jlambda_def}
\end{eqnarray}
Maximizing $J_{\lambda}$ maximizes a fidelity threshold that is guaranteed to be exceeded by at least a fraction $\lambda^{2}/(1+\lambda^{2})$ of successful runs, irrespective of the detailed shape of the fidelity distribution. Thus, $\lambda$ is not merely a Lagrange multiplier; it can be chosen operationally according to the desired confidence. For example, demanding that at least $90\%$ of successful runs exceed a certified threshold corresponds to $\lambda=\sqrt{0.9/0.1}=3$.

When heralding is present, one can combine {\bf Theorem~\ref{thm:Cantelli_Jlambda}} with the ensemble success probability $P_{\mathrm{succ}}$. Let $Y$ denote the event that a given trial is successful and has fidelity at least $J_{\lambda}$. Then,
\begin{eqnarray}
\Pr(Y) = P_{\mathrm{succ}} \Pr\left[X\ge J_{\lambda}\right] \ge P_{\mathrm{succ}}\frac{\lambda^{2}}{1+\lambda^{2}}.
\label{eq:throughput_certified}
\end{eqnarray}
Eq.~(\ref{eq:throughput_certified}) makes explicit why the triple $(F, D, P_{\mathrm{succ}})$ is the natural design space: the post-selected improvements may raise $F$ but can simultaneously increase $D$ and suppress $P_{\mathrm{succ}}$, thereby reducing the rate of producing reliably good outputs.

These considerations suggest the following robust design principles for the conditional CV teleportation. Rather than maximizing $F$ alone, one can (i) maximize $J_{\lambda}$ in Eq.~(\ref{eq:Jlambda_def}) for a chosen confidence parameter $\lambda$, or (ii) solve constrained optimizations of the form
\begin{eqnarray}
\max F ~~\text{subject to}~~ (D \le D_{\max}, ~P_{\mathrm{succ}} \ge P_{\min}),
\label{eq:constrained_design}
\end{eqnarray}
depending on whether the application is limited primarily by worst-case reliability (a $D$ budget), by repetition overhead (a $P_{\mathrm{succ}}$ budget), or by both. In teleportation-based networks, the constraints of the form in Eq.~(\ref{eq:constrained_design}) are often more faithful to operational needs than any mean-only benchmarks: the repeated use of teleportation can tolerate a moderate drop in $F$ if it dramatically reduces the fraction of anomalously poor instances, i.e., if it reduces $D$.

In the MB-NLA setting of Sec.~\ref{sec:mbnla}, the filter parameters (gain $g$ and cut-off $m_{c}$) provide precisely the knobs that tune this robustness trade-off. The large cut-offs make the protocol closer to the deterministic BK baseline, maintaining universality (small $D$), but limiting conditional improvements. The tight cut-offs can raise $F$ near the origin, but at the cost of increased selectivity, as certified by our earlier results ({\bf Theorem~\ref{thm:local_universality_cost}} and {\bf Proposition~\ref{prop:cutoff_selectivity}}) and captured quantitatively by the degradation of $J_{\lambda}$ at large $\lambda$. From this perspective, the purpose of $D$ is not to penalize the non-Gaussianity, but to calibrate how much the apparent improvement in $F$ is driven by input-selective filtering rather than by a uniform enhancement of the effective teleportation channel.

\section{Conclusion}\label{sec:conclusion}

In this work, we have argued that, for the coherent-state benchmarking of CV teleportation, the fidelity deviation $D$ should be read primarily as a universality (input-uniformity) diagnostic, and we made this statement precise by proving that displacement covariance of the effective teleportation channel forces the single-shot fidelity profile to be flat and hence $D=0$ ({\bf Theorem~\ref{thm:covariance_implies_D0}}). We then established a complementary quantitative statement: whenever an apparent fidelity gain is generated through a nonuniform first-order response across the input ensemble, the deviation necessarily increases and, in the perturbative regime, grows linearly with the fidelity gain with a slope fixed by the variance of the selectivity profile ({\bf Theorem~\ref{thm:local_universality_cost}}). Together with the throughput bound in {\bf Theorem~\ref{thm:throughput_bound}}, this identifies $D$ as the price paid for the conditional improvement when one demands reliably good outputs rather than a high mean alone. Applying the framework to MB-NLA-enhanced teleportation~\cite{Chrzanowski2014,Zhao2023,Fiurasek2024}, we explicitly mapped the trade-off geometry in $(F, D, P_{\mathrm{succ}})$ and, by contrasting it with deterministic BK teleportation using non-Gaussian resources, such as photon-subtracted/added entanglement~\cite{Opatrny2000,DellAnno2007,Patra2022}, we clarified that it is a conditional filtering---not resource non-Gaussianity by itself---that generically drives the universality cost.

Beyond the particular case study, the overall message is methodological. Average fidelity $F$ remains an essential benchmark for unknown inputs, but it is incomplete whenever teleportation is used as a composable primitive. A protocol can increase $F$ by selectively concentrating good events and yet become less reliable across the ensemble, which is precisely what $D$ captures. From this perspective, the purpose of introducing $D$ is not to penalize the non-Gaussianity, but to separate ``uniform channel improvement'' from ``input-selective post-selection advantage'' within a single operational diagram. This separation also suggests a natural optimization principle for probabilistic enhancements: rather than tuning the parameters solely to maximize $F$, one can optimize robustness-aware objectives, such as $J_{\lambda} = F - \lambda D$, or equivalently maximize $F$ subject to explicit universality and throughput constraints $(D \le D_{\max}, P_{\mathrm{succ}} \ge P_{\min})$, depending on the architecture-level requirements.

Several extensions follow naturally. First, while we focused on the coherent-state ensembles with a Gaussian prior as the cleanest setting in which displacement covariance translates directly into universality, the same symmetry-first logic can be ported to other families of inputs (e.g., squeezed ensembles) by identifying the relevant group action and the corresponding notion of covariance. Second, the present framework invites a principled experimental reporting standard for conditional teleportation: alongside $(F, P_{\mathrm{succ}})$, one should report $D$ (or equivalently a selectivity index) to quantify how much of the improvement is driven by ensemble distortion, thereby making the ``universality cost'' visible in a device- and application-relevant way. We expect that these refinements will be particularly valuable in the networked and modular CV architectures, where the teleportation is repeatedly invoked and where rare but strongly input-dependent failures can dominate the end-to-end behavior.

\section*{Acknowledgments}

This work was supported by the National Research Foundation of Korea through grants funded by the Ministry of Science, ICT and Future Planning (MSIP) (RS-2024-00432214, RS-2025-03532992, and RS-2025-18362970); by the Institute of Information and Communications Technology Planning and Evaluation through a grant funded by the Korean government (RS-2019-II190003, ``Research and Development of Core Technologies for Programming, Running, Implementing and Validating of Fault-Tolerant Quantum Computing System''). This work is also supported by the Grant No.~K25L5M2C2 at the Korea Institute of Science and Technology Information (KISTI). We acknowledge the Yonsei University Quantum Computing Project Group for providing support and access to the Quantum System One (Eagle Processor), which is operated at Yonsei University.


\bibliography{CV_nongaussian_fidelD}

\end{document}